\documentclass[final,5p,twocolumn,10pt,sort&compress]{elsarticle}

\usepackage[english]{babel}
\usepackage{enumerate}  
\usepackage[shortlabels]{enumitem} 
\usepackage{amsmath, amssymb, amsthm}
\usepackage{empheq} 
\usepackage{graphicx}
\usepackage[dvipsnames]{xcolor} 
\usepackage[colorlinks=true,linkcolor=NavyBlue,citecolor=NavyBlue,urlcolor=NavyBlue]{hyperref} 
\usepackage{appendix}
\usepackage{comment}
\usepackage{subcaption}
\usepackage{booktabs,dcolumn,caption}
\usepackage{tabularx}
\usepackage{csquotes}
\usepackage{multicol} 
\usepackage{balance}
\usepackage{xurl}
\usepackage[most]{tcolorbox}
\usepackage{multirow}
\usepackage{siunitx}
\usepackage{titlesec}
\newtheorem{proposition}{Proposition}
\newtheorem{theorem}{Theorem} 
\newtheorem{definition}{Definition}
\newtheorem{corollary}{Corollary}

\tcbuselibrary{theorems}
\newtcbtheorem{mybox}{Box}%
{colframe=black!50, colback=gray!5, boxrule=0.8pt, fonttitle=\bfseries}{box}

\makeatletter
\def\ps@pprintTitle{%
 \let\@oddhead\@empty 
 \let\@evenhead\@empty
 \def\@oddfoot{}%
 \let\@evenfoot\@oddfoot}
 \makeatother
 
\titleformat{\section}
  {\large\bfseries}        % larger and bold
  {}                       % no numbering (since you use *)
  {0pt}
  {}
\titlespacing*{\section}{0pt}{2.5ex}{1.2ex}

\titleformat{\subsection}
  {\normalsize\bfseries}   % slightly smaller
  {}
  {0pt}
  {}
\titlespacing*{\subsection}{0pt}{1.8ex}{0.8ex}

\titleformat{\subsubsection}
  {\small\itshape}         % lighter + italic → clear contrast
  {}
  {0pt}
  {}
\titlespacing*{\subsubsection}{0pt}{1.2ex}{0.5ex}

\begin{document}
\hypersetup{ % overwrite hyperlink settings
  colorlinks=true,
  linkcolor=MidnightBlue,
  citecolor=MidnightBlue,
  urlcolor=MidnightBlue
}
\begin{frontmatter}

% Title, authors, affiliations.
%%%%%%%%%%%%%%%%%%%%%%%%%%%%%%%%%%%%%%%%%%%%%%%%%%%%%%%%%%%%%%%%%%%%%%%%%%%%%%%%%%%%%%%%%%%%%%%%%%%%%%%%%%%%%%
\title{Strategic Network Abandonment}

\author[1,2]{Sandro Claudio Lera\corref{cor1}}
\author[3]{Andreas Haupt}

\cortext[cor1]{corresponding author: \texttt{leras@sustech.edu.cn}}
\address[1]{Institute of Risk Analysis, Prediction and Management, Southern University of Science and Technology, Shenzhen, China}
\address[2]{Connection Science, Massachusetts Institute of Technology, Cambridge, USA}
\address[3]{Stanford Institute for Human-Centered Artificial Intelligence, Stanford, USA}

% Abstract and keywords. 
%%%%%%%%%%%%%%%%%%%%%%%%%%%%%%%%%%%%%%%%%%%%%%%%%%%%%%%%%%%%%%%%%%%%%%%%%%%%%%%%%%%%%%%%%%%%%%%%%%%%%%%%%%%%%%
\begin{abstract}
Socio-economic networks, from cities and firms to collaborative projects, often appear resilient for long periods before experiencing rapid, cascading decline as participation erodes.
We explain such dynamics through a framework of strategic network abandonment, in which interconnected agents choose activity levels in a network game and remain active only if participation yields higher utility than an improving outside option.
As outside opportunities rise, agents exit endogenously, triggering equilibrium readjustments that may either dissipate locally or propagate through the network.
The resulting decay dynamics are governed by the strength of strategic complementarities, measuring how strongly an agent's incentives depend on the actions of others.
When complementarities are weak, decay follows a heterogeneous threshold process analogous to bootstrap percolation: failures are driven by local neighborhoods, vulnerable clusters can be identified ex ante, and large cascades emerge only through bottom-up accumulation of fragility.
When complementarities are strong, departures propagate globally, producing rupture-like dynamics characterized by metastable plateaus, abrupt system-wide collapse, and limited predictive power of standard spectral or structural indicators.
The comparative effectiveness of intervention depends on the strength of complementarity as well: supporting central agents is most effective under strong complementarities, whereas targeting marginal agents is essential when complementarities are weak.
Together, our results reveal how outside options, network structure, and strategic interdependence jointly determine both the fragility of socio-economic networks and the policies required to sustain them.
\end{abstract}


\end{frontmatter}

\begin{quote}
    \enquote{How did you go bankrupt?} 
    \enquote{Two ways. Gradually, then suddenly.}
    --- Ernest Hemingway
\end{quote}

%%%%%%%%%%%%%%%%%%%%%%%%%%%%%%%%%%%%%%%%%%%%%%%%%%%%%%%%%%%%%%%%%%%%%%%%%%%%%%%%%%%%%%%%%%%%%%%%%%%%%%%%%%%%%%
%%%%%%%%%%%%%%%%%%%%%%%%%%%%%%%%%%%%%%%%%%%%%%%%%%%%%%%%%%%%%%%%%%%%%%%%%%%%%%%%%%%%%%%%%%%%%%%%%%%%%%%%%%%%%%
\section*{Introduction}
%%%%%%%%%%%%%%%%%%%%%%%%%%%%%%%%%%%%%%%%%%%%%%%%%%%%%%%%%%%%%%%%%%%%%%%%%%%%%%%%%%%%%%%%%%%%%%%%%%%%%%%%%%%%%%
%%%%%%%%%%%%%%%%%%%%%%%%%%%%%%%%%%%%%%%%%%%%%%%%%%%%%%%%%%%%%%%%%%%%%%%%%%%%%%%%%%%%%%%%%%%%%%%%%%%%%%%%%%%%%%

The abandonment of socio-economic systems often begins slowly and then accelerates rapidly.
Online communities, collaborative projects, firms, and other participation-based networks can experience long periods of declining engagement, during which individual departures appear inconsequential, followed by sudden cascades that rapidly dismantle collective activity.
Understanding when gradual erosion escalates into abrupt collapse is a central challenge for predicting the stability of interconnected social and economic systems \cite{Glaeser2005}.

Decay processes in complex systems are often analyzed through models of failure propagation, in which the loss of individual nodes or links can trigger cascades whose extent depends on network connectivity and agent-level heterogeneity, such as differences in degree, thresholds, or vulnerability \cite{Granovetter1978,Watts2002,Morris2000}.
This perspective has been especially influential in the study of infrastructure systems, such as power grids or communication networks, where random breakdowns or targeted attacks redistribute load and can precipitate large-scale failures \cite{Albert2000,Cohen2000,Cohen2001}.
In many socio-economic and social networks, however, decline follows a different logic.
Participation erodes not because components fail mechanically or because failure spreads infectiously, but because agents make strategic decisions to disengage when continued participation becomes less attractive than available outside opportunities.
As a result, departures are driven by incentives rather than exogenous shocks, and the resulting dynamics reflect endogenous feedback between individual choices and network structure.

Work in economics and sociology has long emphasized that participation in networks is often a matter of strategic choice \cite{Jackson2002,Bramoulle2009,CalvoArmengol2009,Aral2012}.
Network games show that equilibrium activity levels depend strongly on agents’ positions in the network, with outcomes closely linked to measures such as Bonacich centrality \cite{Ballester2006,Bramoulle2007,Bramoulle2014,Galeotti2010}.
However, these approaches typically focus on static equilibria rather than the dynamic unraveling of networks through iterative exit.
At the same time, empirical studies of organizational decline and community disengagement document how the departure of some participants can reduce the incentives of those who remain, potentially triggering rapid collective withdrawal \cite{Torok2017,Avelino2019}.

To account for such dynamics, we study the strategic abandonment of socio-economic networks governed by economically motivated agents.
We model a repeated network game in which agents remain active only if their equilibrium payoff from participation exceeds that of an outside option that improves over time.
Agents choose activity levels through a linear--quadratic network game, yielding equilibrium actions proportional to Bonacich centrality \cite{Ballester2006}.
As outside opportunities rise, agents at the margin of participation exit first.
Their departure reduces the equilibrium payoffs of remaining agents, potentially triggering further exits.
Repeated iteratively, this process produces either slow, incremental decline or sudden cascades that rapidly dismantle the network.

Our analysis shows that the resulting decay dynamics are governed by the strength of strategic complementarities.
When complementarities are weak, departures propagate primarily through local neighborhood structure and the dynamics reduce to a generalized threshold process closely related to heterogeneous $k$-core percolation.
When complementarities are strong, departures propagate globally through equilibrium feedback effects, producing abrupt system-wide collapse and substantial sensitivity to small structural perturbations.
These regimes also imply different intervention strategies:
supporting vulnerable peripheral agents is most effective when complementarities are weak, whereas targeting central agents becomes optimal when complementarities are strong and spillovers diffuse broadly through the network.

%%%%%%%%%%%%%%%%%%%%%%%%%%%%%%%%%%%%%%%%%%%%%%%%%%%%%%%%%%%%%%%%%%%%%%%%%%%%%%%%%%%%%%%%%%%%%%%%%%%%%%%%%%%%%%
%%%%%%%%%%%%%%%%%%%%%%%%%%%%%%%%%%%%%%%%%%%%%%%%%%%%%%%%%%%%%%%%%%%%%%%%%%%%%%%%%%%%%%%%%%%%%%%%%%%%%%%%%%%%%%
\section*{Results}
%%%%%%%%%%%%%%%%%%%%%%%%%%%%%%%%%%%%%%%%%%%%%%%%%%%%%%%%%%%%%%%%%%%%%%%%%%%%%%%%%%%%%%%%%%%%%%%%%%%%%%%%%%%%%%
%%%%%%%%%%%%%%%%%%%%%%%%%%%%%%%%%%%%%%%%%%%%%%%%%%%%%%%%%%%%%%%%%%%%%%%%%%%%%%%%%%%%%%%%%%%%%%%%%%%%%%%%%%%%%%

%%%%%%%%%%%%%%%%%%%%%%%%%%%%%%%%%%%%%%%%%%%%%%%%%%%%%%%%%%%%%%%%%%%%%%%%%%%%%%%%%%%%%%%%%%%%%%%%%%%%%%%%%%%%%%
\subsection*{A Model for Strategic Network Abandonment}
%%%%%%%%%%%%%%%%%%%%%%%%%%%%%%%%%%%%%%%%%%%%%%%%%%%%%%%%%%%%%%%%%%%%%%%%%%%%%%%%%%%%%%%%%%%%%%%%%%%%%%%%%%%%%%

\subsubsection*{Single-Step Equilibrium}
%%%%%%%%%%%%%%%%%%%%%%%%%%%%%%%%%%%%%%%%%%%%%%%%%%%%%%%%%%%%%%%%%%%%%%%%%%%%%%%%%%%%%%%%%%%%%%%%%%%%%%%%%%%%%%

We consider a network consisting of $N$ agents connected by weighted, undirected links. 
The structure of the network is described by its adjacency matrix $A$, where $A_{ij}>0$ if there is a link between agent $i$ and $j$, and $A_{ij}=0$ otherwise.

Each agent $i$ chooses a scalar activity level $x_i \geqslant 0$, representing how much effort to invest in the system.
Agent $i$'s utility is given by a standard linear-quadratic form \cite{Jackson2002,Bramoulle2009,CalvoArmengol2009,Aral2012}:
\begin{equation}
	\label{eq:linquad_utility}
	U_i = \alpha x_i - \frac{1}{2} x_i^2 +  \beta x_i \sum_{j=1}^N A_{ij} x_j.
\end{equation}
This linear-quadratic structure is widely used because it arises as a first-order approximation to a broad class of network interaction games, 
making it a natural and highly general framework for modeling strategic behavior in socio-economic systems. 
The first term captures the direct utility from agent $i$'s own effort, while the second term represents the cost of exerting effort.  
The final term captures the benefit an agent receives from the effort of its neighbors, scaled by its own activity, scaled by the complementarity parameter $\beta \geqslant 0$.  
Larger values of $\beta$ imply stronger strategic complementarities, meaning an agent gains more from being connected to active neighbors.

Given a fixed network $A$, the vector of equilibrium activity levels $x = (x_1, \dots, x_N)$ is determined by simultaneous best responses.  
Solving the first-order conditions yields a unique Nash equilibrium provided that $\beta$ is sufficiently small so that the matrix $I - \beta A$ is invertible, which occurs whenever $\beta \rho(A) < 1$, where $\rho(A)$ is the spectral radius of $A$.  
The equilibrium actions are given by 
\begin{equation}
	\label{eq:equilibrium_action}
	x^*(A) = \alpha (I - \beta A)^{-1} \mathbf{1},	
\end{equation}
where $\mathbf{1}$ is a vector of ones.
The term $(I - \beta A)^{-1} \mathbf{1}$ is known as the \emph{Bonacich centrality} of the network.  
It generalizes degree centrality by accounting for both direct and indirect connections, with longer paths discounted by powers of $\beta$.  
Agents with high Bonacich centrality are strongly connected to others who are themselves well connected.

At equilibrium, the utility of agent $i$ is
\begin{equation}
	\label{eq:equilibrium_utility}
	U_i^\ast(x_i) = \frac{1}{2} \left( x_i^* \right)^2,
\end{equation}
so ranking agents by equilibrium utility is equivalent to ranking them by their activity level $x_i$.  
Thus, in what follows, we often use the terms \enquote{utility} and \enquote{centrality} interchangeably.

\subsubsection*{Iterative Removal Process}
%%%%%%%%%%%%%%%%%%%%%%%%%%%%%%%%%%%%%%%%%%%%%%%%%%%%%%%%%%%%%%%%%%%%%%%%%%%%%%%%%%%%%%%%%%%%%%%%%%%%%%%%%%%%%%

%%%%%%%%%%%%%%%%%%%%%%%%%%%%%%%%%%%%%%%%%%%%%%%%%%%%%%%%%%%%%%%%%%%%%%%%%%%%%%%%%%%%%%%%%%%%%%%%%%%%%%%%%%%%%%
\begin{figure*}[!htb]
    \centering
    \includegraphics[width=\linewidth]{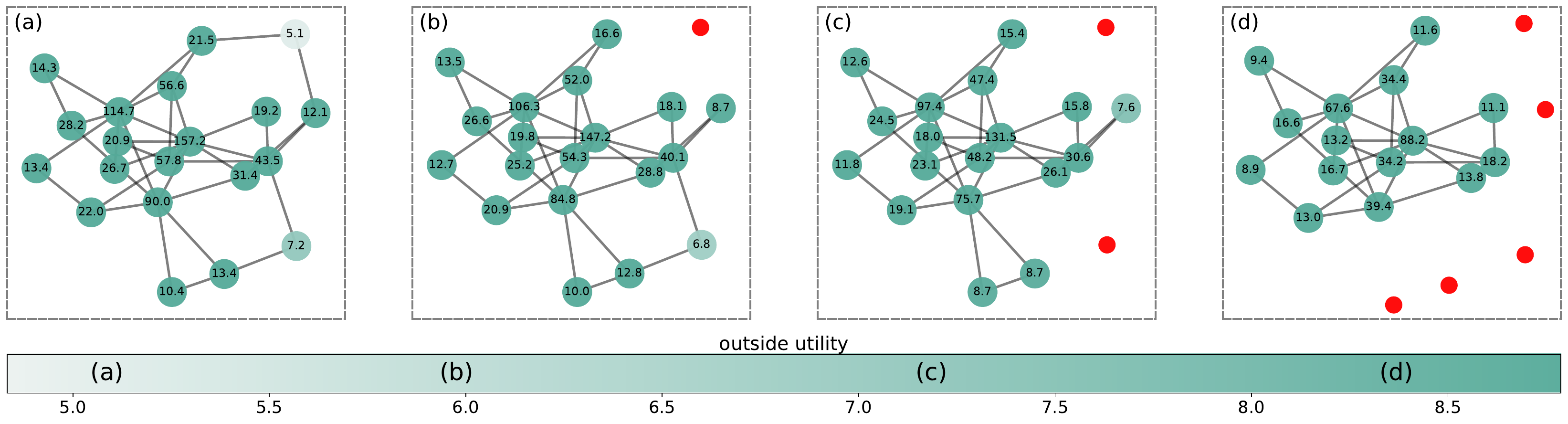}
    \caption{
    \textbf{The Strategic Network Abandonment Model.}
	(a) The initial state of a network of 20 agents (nodes), 
    with node sizes proportional to payoff utility \eqref{eq:equilibrium_utility} from game \eqref{eq:linquad_utility} with $\alpha=1$, $\beta=0.3$ and $\beta \rho(A) = 0.9$. 
	The outside utility is set equal to $5.05$ and hence below the minimal utility of $5.1$ that the agent with the lowest utility gets for being in the network. 
	(b) The outside option is increased, causing the lowest-utility agent to exit without triggering further departures.
     (c) A further increase in the outside option leads to the exit of one additional agent, again without a cascade.
    (d) A further increase of the outside option triggers a cascade, 
    where the abandonment of a third agent induces the abandonment of two more agents.
    We refer to Box~\ref{box:iterative-algorithm} for details. 
    }
    \label{fig:decay_concept}
\end{figure*}
%%%%%%%%%%%%%%%%%%%%%%%%%%%%%%%%%%%%%%%%%%%%%%%%%%%%%%%%%%%%%%%%%%%%%%%%%%%%%%%%%%%%%%%%%%%%%%%%%%%%%%%%%%%%%%

We now study how the network evolves when some agents abandon the system over time.  
We begin at time $t=0$ with an initial network $A^{(0)}$.  
The corresponding equilibrium activity levels are
$
x^{(0)} \equiv x^*\big(A^{(0)}\big),
$
and the initial utilities are
$
U^{(0)}_i \equiv U^*\big(x^{(0)}_i\big),
$
for each agent $i$.

To model the possibility that agents leave when their utility is too low, we introduce an \emph{outside utility} $u_{\mathrm{out}}^{(t)}$.  
This represents the minimum utility an agent must achieve to justify staying in the system, with lower-utility agents exiting to pursue outside opportunities.  
Examples include 
users disengaging from an online platform in favor of alternative services, 
contributors leaving a collaborative project as their engagement declines, 
or participants shifting attention to competing activities that offer higher returns.

At the initial stage, we gradually raise the outside utility until it equals the minimum utility among the currently active agents:
$
u_{\mathrm{out}}^{(0)} = \min_{i} U^{(0)}_i.
$
By construction, at least one agent, and possibly several, now has utility exactly equal to $u_{\mathrm{out}}^{(0)}$.  
These agents exit the system, and we remove them from the network.

Removing these agents changes the equilibrium of the remaining network, because their connections and payoffs are altered.  
We therefore recompute the equilibrium activities and utilities for the reduced network, holding the outside utility fixed at the current value $u_{\mathrm{out}}^{(0)}$.  
If, after this recalculation, additional agents now have utilities below $u_{\mathrm{out}}^{(0)}$, they too will exit.  
We repeat this process until no further agents leave at this fixed threshold.
At that point, the set of remaining agents is labeled as the next-stage network $A^{(1)}$, and the procedure continues with $t=1$.
This procedure is repeated until no agent remains in the system.  
At each stage, we wait for the network to fully re-equilibrate before increasing the outside utility further.
This \emph{adiabatic approximation} ensures that the network is always in a steady state before new external pressures---that is, the outside option---are increased.

This iterative removal process generates a sequence of progressively smaller networks
$
A^{(0)} \supset A^{(1)} \supset A^{(2)} \supset \cdots,
$
together with an increasing sequence of outside utility levels
$
u_{\mathrm{out}}^{(0)} < u_{\mathrm{out}}^{(1)} < u_{\mathrm{out}}^{(2)} < \cdots,
$
until all agents have exited the system.
Box~\ref{box:iterative-algorithm} summarizes our set-up and Figure \ref{fig:decay_concept} visualizes the concept. 
We refer to the SI Appendix for a proof that this process is well-defined, meaning that the final outcome does not depend on the order in which agents are removed during a cascade and that a unique stable network remains at each stage.
Going forward, we omit the superscript $(t)$ when it is clear from context that the object refers to the network at some generic time step.

\subsubsection*{The Process is Well Defined}
%%%%%%%%%%%%%%%%%%%%%%%%%%%%%%%%%%%%%%%%%%%%%%%%%%%%%%%%%%%%%%%%%%%%%%%%%%%%%%%%%%%%%%%%%%%%%%%%%%%%%%%%%%%%%%

The iterative removal process raises a nontrivial issue: when a cascade occurs, the final survivor set is not obviously unique.
Once some agents leave, the utilities of the remaining agents decrease, potentially triggering additional departures.
Different sequences of exits could therefore, in principle, produce different surviving subnetworks.
In the Methods and Appendix, we show that strategic complementarities imply a monotonicity property under which the union of stable subnetworks remains stable.
As a consequence, there exists a unique maximal stable set for every outside-utility level, implying that the final cascade outcome is independent of the order in which unstable agents are removed.

The process also admits a dynamic interpretation.
Consider agents who repeatedly decide whether to remain in the network or permanently leave as outside opportunities gradually improve over time.
Although our removal procedure updates the network one step at a time, we show in the Methods and Appendix that it is fully consistent with forward-looking strategic behavior.
An agent whose utility inside the network still exceeds the outside option has no incentive to leave immediately, since exiting remains possible later if conditions deteriorate further.
Under these conditions, the iterative removal process traces the largest self-sustaining network that can persist at each level of external pressure.

%%%%%%%%%%%%%%%%%%%%%%%%%%%%%%%%%%%%%%%%%%%%%%%%%%%%%%%%%%%%%%%%%%%%%%%%%%%%%%%%%%%%%%%%%%%%%%%%%%%%%%%%%%%%%%
\begin{mybox}{Strategic Network Abandonment}{iterative-algorithm}
\begin{enumerate}[(i)]

    \item 
    At time $t$, start with network $A^{(t)}$ and compute equilibrium activities $x^{(t)}$ and utilities $U^{(t)}_i$ via Eq. \eqref{eq:equilibrium_action} and Eq. \eqref{eq:equilibrium_utility}, respectively.
    
    \item 
    Set the outside utility to the current minimum:
    $
    u_{\mathrm{out}}^{(t)} = \min_i U^{(t)}_i.
    $
    
    \item 
    \label{its:keep_removing}
    Remove all agents with $U_i \leqslant u_{\mathrm{out}}^{(t)}$ and recompute the equilibrium on the remaining subset.
    
    \item 
    Repeat step \ref{its:keep_removing} until no further agents fall below $u_{\mathrm{out}}^{(t)}$.  
    The set of agents removed during this stage constitutes the \emph{cascade} triggered by $u_{\mathrm{out}}^{(t)}$.
    
    \item 
    Once the cascade stops, define the surviving network as $A^{(t+1)}$ and proceed to the next stage $t \mapsto t+1$.
    
\end{enumerate}
\end{mybox}
%%%%%%%%%%%%%%%%%%%%%%%%%%%%%%%%%%%%%%%%%%%%%%%%%%%%%%%%%%%%%%%%%%%%%%%%%%%%%%%%%%%%%%%%%%%%%%%%%%%%%%%%%%%%%%

%%%%%%%%%%%%%%%%%%%%%%%%%%%%%%%%%%%%%%%%%%%%%%%%%%%%%%%%%%%%%%%%%%%%%%%%%%%%%%%%%%%%%%%%%%%%%%%%%%%%%%%%%%%%%%
\subsection*{Local and Global Decay Dynamics}
%%%%%%%%%%%%%%%%%%%%%%%%%%%%%%%%%%%%%%%%%%%%%%%%%%%%%%%%%%%%%%%%%%%%%%%%%%%%%%%%%%%%%%%%%%%%%%%%%%%%%%%%%%%%%%

We now turn to the analytical characterization of the iterative decay process described above.
Our goal is to understand under what conditions a small initial set of departures remains localized or, alternatively, triggers a global cascade that ultimately results in the full abandonment of the network.

\subsubsection*{General Problem Formulation}
%%%%%%%%%%%%%%%%%%%%%%%%%%%%%%%%%%%%%%%%%%%%%%%%%%%%%%%%%%%%%%%%%%%%%%%%%%%%%%%%%%%%%%%%%%%%%%%%%%%%%%%%%%%%%%

Recall that equilibrium activities are given by Eq. \eqref{eq:equilibrium_action} so each agent's utility is influenced not only by its direct neighbors but also by all indirect connections, weighted by powers of $\beta$.
When some agents are removed, the entire inverse matrix $(I - \beta A)^{-1}$ changes, implying that every agent's utility may shift, even if it is far from the initially removed agent.
This makes the decay process non-local. 
This can be seen by expanding the inverse
\begin{equation} \label{eq:neumann}
(I - \beta A)^{-1} = \sum_{k=0}^{\infty} \beta^k A^k,
\end{equation}
which reveals that each term $A^k$ counts all walks of length $k$ among the surviving agents,  
weighted by the factor $\beta^k$ that down-weights longer paths.

%%%%%%%%%%%%%%%%%%%%%%%%%%%%%%%%%%%%%%%%%%%%%%%%%%%%%%%%%%%%%%%%%%%%%%%%%%%%%%%%%%%%%%%%%%%%%%%%%%%%%%%%%%%%%%
\begin{figure*}[!htb]
    \centering
    \includegraphics[width=\linewidth]{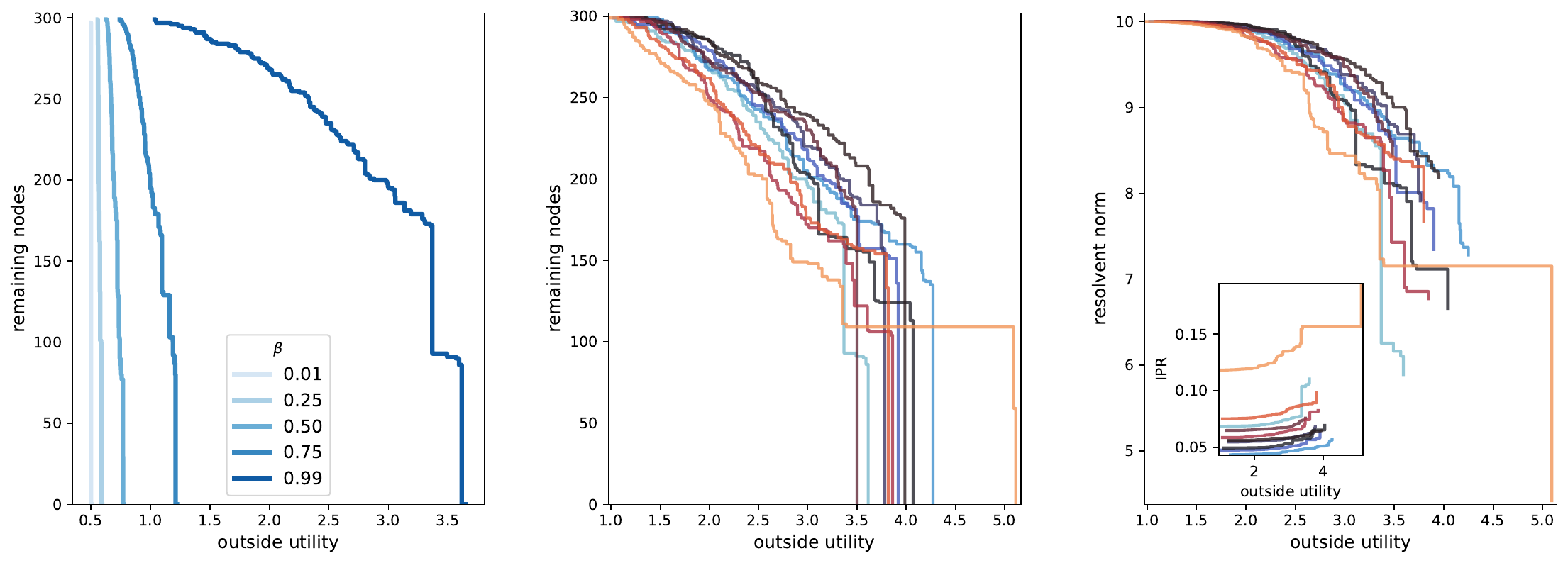}
	\caption{
	\textbf{Simulated network decay.}
	We simulate the abandonment dynamics as the outside utility increases in networks drawn from a power-law cluster ensemble with $N=300$ agents, $m=2$, and $p=0.1$.
	(\textit{Left})
	Decay dynamics in a fixed network realization (normalized to $\rho(A)=1$) under different levels of strategic complementarity $\beta$.
	For small $\beta$, decay proceeds in abrupt, punctuated steps.
	For larger $\beta$, stronger feedback effects propagate shocks more broadly, leading to a smoother and more continuous decline in network size.
	(\textit{Middle})
	Decay across $10$ independent ensemble realizations with $\beta \rho(A)=0.9$.
	In all cases, the network shrinks gradually before undergoing a sudden collapse.
	However, the collapse point varies substantially across realizations, indicating strong sensitivity to small structural differences.
	(\textit{Right})
	Evolution of the resolvent norm and the inverse participation ratio (IPR) of the leading eigenvector for the same realizations as in the middle panel.
	Both measures evolve smoothly and exhibit no clear precursors to the eventual collapse, reflecting the absence of structural reorganization prior to failure.
	}
    \label{fig:beta_and_stochasticity_dependence}
\end{figure*}
%%%%%%%%%%%%%%%%%%%%%%%%%%%%%%%%%%%%%%%%%%%%%%%%%%%%%%%%%%%%%%%%%%%%%%%%%%%%%%%%%%%%%%%%%%%%%%%%%%%%%%%%%%%%%%

\subsubsection*{Approximate Local Dynamics}
%%%%%%%%%%%%%%%%%%%%%%%%%%%%%%%%%%%%%%%%%%%%%%%%%%%%%%%%%%%%%%%%%%%%%%%%%%%%%%%%%%%%%%%%%%%%%%%%%%%%%%%%%%%%%%

The problem simplifies considerably when $\beta$ is small relative to network connectivity.
As shown in the Methods section, in this limit an agent exits once it has lost sufficiently many immediate neighbors.
Specifically, each agent $j$ is associated with an integer threshold
\begin{equation}
	\label{eq:integer_threshold}
	\theta_j = \left\lceil \frac{ U_j - u_\text{out}^{(t)} }{\alpha^2 \beta} \right\rceil,
\end{equation}
and fails whenever the number of failed neighbors meets or exceeds $\theta_j$.

This small-$\beta$ approximation is mathematically equivalent to a \emph{generalized $k$-core peeling process}, also known from heterogeneous percolation \cite{Watts2002,Gleeson2007,Holroyd2003}.
A key insight is that large-scale cascades require a sufficient density of \emph{vulnerable} agents whose thresholds are low enough to be triggered by a small initial shock.
When such agents are scattered throughout the network, a single failure typically remains localized.
However, if they form a connected cluster that spans a macroscopic fraction of the network, an initial departure can ignite a self-sustaining chain reaction.
This leads to a sharp transition: as the fraction of vulnerable agents or the network connectivity crosses a critical value, the cascade size jumps abruptly from negligible to system-wide.

\subsubsection*{Global Dynamics and Limits of Predictability}
%%%%%%%%%%%%%%%%%%%%%%%%%%%%%%%%%%%%%%%%%%%%%%%%%%%%%%%%%%%%%%%%%%%%%%%%%%%%%%%%%%%%%%%%%%%%%%%%%%%%%%%%%%%%%%

Empirical evidence suggests that the low-$\beta$ approximation is rarely adequate in economic and social network games.
In most applications, individuals’ actions are strongly interdependent, and payoffs exhibit pronounced strategic complementarities rather than weak, local interactions.
Field and laboratory studies of linear–quadratic network games consistently estimate $\beta$ in the range of $0.2$ to $0.8$, implying substantial mutual reinforcement of actions across connected agents \citep{Bramoulle2009,CalvoArmengol2009,Aral2012}.
These magnitudes indicate that real networks typically operate in a high-$\beta$ regime, where behavior depends on direct but also on indirect connections and local shocks propagate widely through feedback effects, yet remain well-defined in the spectral sense, with $\beta\rho(A) \lesssim 1$ ensuring a unique and stable equilibrium.
Accordingly, the small-$\beta$ limit, where utilities depend only on immediate neighbors and cascades follow simple threshold rules, captures only a narrow slice of economically relevant behavior.

To illustrate this, we generate a network of $N = 300$ agents from the power-law cluster (Holme–Kim) ensemble which combines preferential attachment with triadic closure.
Each new agent forms $m = 2$ links to existing agents, and with probability $p = 0.1$ each new edge closes a triangle, producing a network with both a heavy-tailed degree distribution and realistic clustering typical of social and collaboration networks.
Finally, the edge weights are uniformly rescaled such that $\rho(A) = 1$.
We then simulate the network decay with varying $\beta$ from $0.01$ to $0.99$. 
As shown in Figure~\ref{fig:beta_and_stochasticity_dependence} (left), 
for small $\beta$, where utilities depend primarily on local neighbors, the decay proceeds in abrupt, punctuated steps.
When $\beta$ is larger, in contrast, indirect feedbacks diffuse the impact of each agent’s departure more widely across the network, 
leading to a more gradual reduction in overall activity.

This dependence on long-range reinforcement mirrors the behavior of globally coupled failure models, such as \emph{fiber-bundle models} with global load sharing \cite{Sethna1993,Johansen2000,Alava2006,Pradhan2010}.  
In these systems, failures redistribute load across all surviving components, generating a metastable regime that persists until a global stability condition is violated.  
Collapse then occurs abruptly and system-wide, without early-warning signals typical of second-order transitions, and with substantial variability in the critical point across realizations.

A similar pattern arises in our strategic network abandonment model.
Exiting agents trigger equilibrium responses throughout the network, leading to a metastable phase of gradual shrinkage followed by a sudden, discontinuous collapse.
The combination of abruptness and sensitivity to fine structural differences makes the timing of collapse inherently difficult to predict in the high-$\beta$ regime.

Figure~\ref{fig:beta_and_stochasticity_dependence} (middle and right) illustrates this behavior across $10$ realizations of the same network ensemble. 
The number of remaining agents declines gradually before collapsing at different outside options across realizations, despite identical model parameters.  
The resolvent norm $\|(I - \beta A)^{-1}\|_2$, which serves as a measure of global susceptibility, decreases smoothly and provides no advance signal of the impending rupture.  
Likewise, the inverse participation ratio (IPR) of the leading eigenvector remains nearly flat, exhibiting only a sharp spike at collapse. 
Thus, although the dynamics are deterministic conditional on the network, small structural differences generate substantial variability in collapse timing.
We refer to the SI Appendix for an exact sufficient condition characterizing global collapse.

%%%%%%%%%%%%%%%%%%%%%%%%%%%%%%%%%%%%%%%%%%%%%%%%%%%%%%%%%%%%%%%%%%%%%%%%%%%%%%%%%%%%%%%%%%%%%%%%%%%%%%%%%%%%%%
\subsection*{Empirical Analysis}
%%%%%%%%%%%%%%%%%%%%%%%%%%%%%%%%%%%%%%%%%%%%%%%%%%%%%%%%%%%%%%%%%%%%%%%%%%%%%%%%%%%%%%%%%%%%%%%%%%%%%%%%%%%%%%

%%%%%%%%%%%%%%%%%%%%%%%%%%%%%%%%%%%%%%%%%%%%%%%%%%%%%%%%%%%%%%%%%%%%%%%%%%%%%%%%%%%%%%%%%%%%%%%%%%%%%%%%%%%%%%
\begin{figure*}[!htb]
    \centering
    \includegraphics[width=\linewidth]{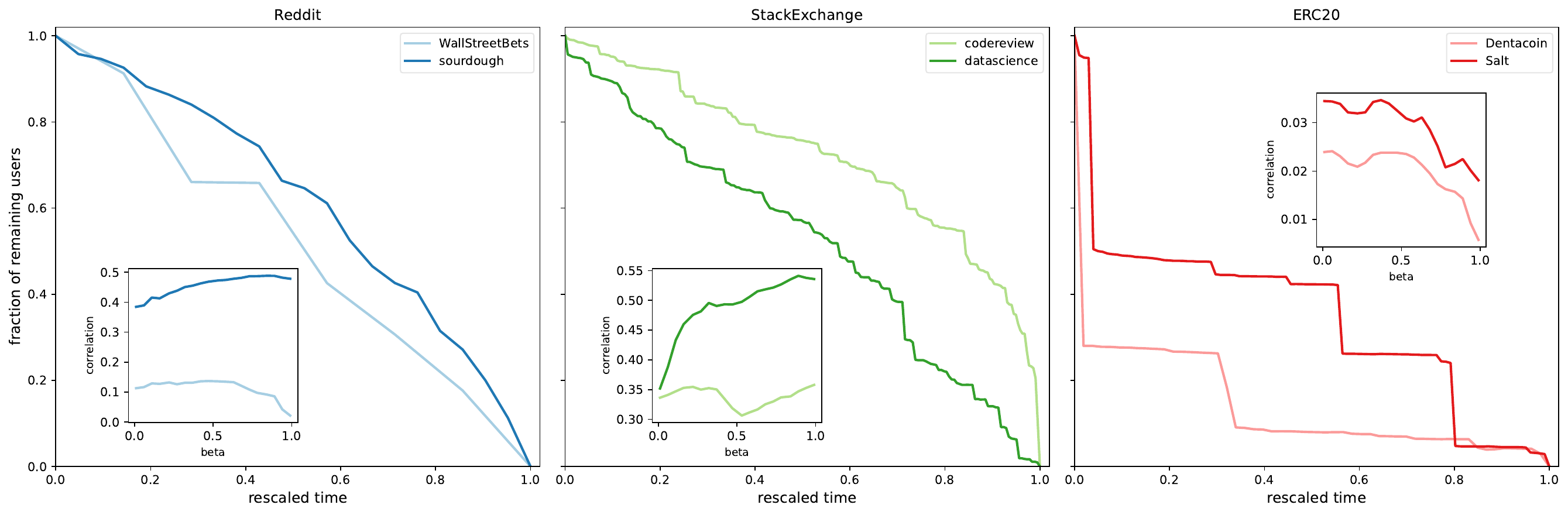}
	\caption{
	\textbf{Decay of real-world networks.}
	We analyze six socio-economic systems: 
	two abandoned ERC20 cryptocurrency projects, 
	two subreddits with declining activity, 
	and two Stack Exchange forums experiencing sustained disengagement.
	(\textit{Top})
	Temporal evolution of participation in each network.
	Each panel shows two representative examples per domain.
	Cryptocurrency projects exhibit relatively rapid collapse, whereas Reddit and Stack Exchange communities tend to decay more gradually.
	(\textit{Insets})
	For each network, we compare the empirical exit ordering of agents to the ordering implied by the model across different values of the complementarity parameter $\beta$.
	We report the resulting rank correlation as a function of $\beta$ and define $\beta^*$ as the value that maximizes this correspondence.
	}
    \label{fig:empirical_abandoment_overview}
\end{figure*}
%%%%%%%%%%%%%%%%%%%%%%%%%%%%%%%%%%%%%%%%%%%%%%%%%%%%%%%%%%%%%%%%%%%%%%%%%%%%%%%%%%%%%%%%%%%%%%%%%%%%%%%%%%%%%%

We now apply our framework to real-world socio-economic networks that exhibit sustained decline in participation.
Our goal is to assess whether the strategic abandonment model can reproduce the observed ordering of agent exits and how its performance compares to standard structural benchmarks.

We consider three classes of networks: Reddit interaction networks, Stack Exchange communication networks, and ERC20 cryptocurrency transaction networks. 
Specifically, we study 
two large subreddits (\textit{WallStreetBets} and \textit{sourdough}) whose activity subsided after the COVID-era surge in online engagement, 
two Stack Exchange communities (\textit{codereview} and \textit{datascience}) that have seen reduced participation with the rise of large language models as alternative knowledge sources,
and two ERC20 projects (\textit{Dentacoin} and \textit{Salt}) that experienced sharp decline following the post-2018 contraction in crypto markets. 
We refer to the Methods section for details. 
These settings are particularly well suited to our framework, as participation is driven by platform-specific incentives and subject to increasingly attractive outside options, making strategic disengagement the central mechanism of decline.
In each case, we construct temporal graphs from time-stamped interaction data by aggregating edges over fixed time intervals and retaining only users that are active over sufficiently long periods.
We focus on the declining phase following peak participation and track the induced sequence of progressively shrinking graphs as activity falls over time.
Figure~\ref{fig:empirical_abandoment_overview} provides an overview of the resulting decay dynamics across these systems. 

To evaluate the model, we compare the empirical order in which agents exit the network to the ordering implied by the strategic abandonment model for different values of the complementarity parameter $\beta$.
Because the equilibrium condition depends on the product $\beta \rho(A)$, edge weights are only identified up to a multiplicative scale.
We therefore normalize the initial adjacency matrix such that $\rho(A)=1$, allowing $\beta\in(0,1)$ to directly represent the effective strength of strategic complementarities.
For each $\beta$, we simulate the decay process and compute the Spearman rank correlation between the empirical and model-implied exit orders.
This yields a mapping from $\beta$ to predictive performance, allowing us to identify the value $\beta^*$ that best aligns the model with the data.
The insets in Figure~\ref{fig:empirical_abandoment_overview} visualize this dependence for each network.

We benchmark our model against three standard centrality-based heuristics computed on the initial network: Bonacich centrality (BC), eigenvector centrality (EV), and $k$-core number (KC).
Each of these induces a predicted exit ordering, which we compare to the empirical ordering using the same Spearman rank-correlation metric.

Table~\ref{tbl:empirical_performance} summarizes the results across all datasets.
A clear pattern emerges.
In systems where the estimated complementarity $\beta^*$ is small (ERC20 networks), 
the predictive performance of the strategic abandonment model is comparable to that of the $k$-core measure, consistent with the theoretical result that the model reduces to a local threshold process in the low-$\beta$ regime.
In contrast, in systems with large $\beta^*$, the model behaves similarly to eigenvector and Bonacich centrality computed on the initial network, consistent with the increasing importance of long-range network amplification in the high-$\beta$ regime.
Most notably, for intermediate values of $\beta^*$, the strategic abandonment model consistently outperforms all benchmarks, including Bonacich centrality.
This highlights the importance of modeling the endogenous re-equilibration of incentives during network decline, rather than relying solely on static centrality measures computed from the initial network.

%%%%%%%%%%%%%%%%%%%%%%%%%%%%%%%%%%%%%%%%%%%%%%%%%%%%%%%%%%%%%%%%%%%%%%%%%%%%%%%%%%%%%%%%%%%%%%%%%%%%%%%%%%%%%%
\begin{table}[t]
\centering
\footnotesize
\setlength{\tabcolsep}{4pt}
\renewcommand{\arraystretch}{0.9}

\begin{tabular}{llc@{\hspace{6pt}}|@{\hspace{6pt}}rrrr}
\toprule
 &  & $\beta^*$ & SA & KC & EV & BC \\
\midrule

\multirow{2}{*}{Reddit}
 & WallStreetBets & 0.47 & 13.7 &  8.4 & 7.5 & 9.7 \\
 & sourdough      & 0.84 & 48.8 & 33.0 & 48.7 & 45.3 \\
\cmidrule(lr){1-7}

\multirow{2}{*}{StackExchange}
 & codereview  & 0.99 & 37.8 & 40.7 & 41.2 & 34.9 \\
 & datascience & 0.89 & 54.1 & 30.9 & 53.6 & 55.7 \\
\cmidrule(lr){1-7}

\multirow{2}{*}{ERC20}
 & Dentacoin & 0.06 & 2.4 & 1.9 & 1.3 & 0.8 \\
 & Salt      & 0.37 & 3.5 & 10.4 & 2.1 & 4.1 \\

\bottomrule
\end{tabular}

\caption{
Comparison of model performance across datasets.
We report $\beta^*$, the value of the complementarity parameter that maximizes the rank correlation between the strategic abandonment model (SA) and the empirical decay.
We further report the corresponding correlations (in percent) for SA and three benchmarks: $k$-core centrality (KC), eigenvector centrality (EV), and Bonacich centrality (BC).
}
\label{tbl:empirical_performance}
\end{table}
%%%%%%%%%%%%%%%%%%%%%%%%%%%%%%%%%%%%%%%%%%%%%%%%%%%%%%%%%%%%%%%%%%%%%%%%%%%%%%%%%%%%%%%%%%%%%%%%%%%%%%%%%%%%%%

%%%%%%%%%%%%%%%%%%%%%%%%%%%%%%%%%%%%%%%%%%%%%%%%%%%%%%%%%%%%%%%%%%%%%%%%%%%%%%%%%%%%%%%%%%%%%%%%%%%%%%%%%%%%%%
\subsection*{How to Sustain a Strategic Network}
%%%%%%%%%%%%%%%%%%%%%%%%%%%%%%%%%%%%%%%%%%%%%%%%%%%%%%%%%%%%%%%%%%%%%%%%%%%%%%%%%%%%%%%%%%%%%%%%%%%%%%%%%%%%%%

\subsubsection*{Top-Down Welfare Maximization}
%%%%%%%%%%%%%%%%%%%%%%%%%%%%%%%%%%%%%%%%%%%%%%%%%%%%%%%%%%%%%%%%%%%%%%%%%%%%%%%%%%%%%%%%%%%%%%%%%%%%%%%%%%%%%%

In many socio-economic systems, it may be desirable to prevent or delay network decay by providing targeted support to selected agents.
Such support can take the form of platform incentives, preferential visibility, moderation or governance privileges, financial rewards, or other mechanisms that increase an agent’s stand-alone incentive to remain active in the system.
Up to this point, the individual incentive parameter $\alpha$ has been assumed homogeneous across agents.
We now relax this assumption and allow for heterogeneous incentives $\alpha_i$, reflecting that some agents may receive greater external support or preferential treatment than others.
In this case, the equilibrium actions in Eq.~\eqref{eq:equilibrium_action} become  
\begin{equation}
x = (I - \beta A)^{-1}\boldsymbol{\alpha},
\end{equation}
with $\boldsymbol{\alpha} = (\alpha_1, \ldots, \alpha_N)$.  

A straightforward approach to designing support policies is to allocate a fixed budget so as to maximize overall welfare in the network.  
This leads to the incentive-targeting problem, where the welfare-maximizing allocation aligns support with the dominant eigenvector of the adjacency matrix $A$ \cite{Galeotti2020}.
Such a policy concentrates resources on the most central and strategically influential agents and can be interpreted as a \emph{top-down} or \emph{trickle-down} strategy in which benefits are expected to diffuse through the network via complementarities.

\subsubsection*{Stabilizing Vulnerable Agents}
%%%%%%%%%%%%%%%%%%%%%%%%%%%%%%%%%%%%%%%%%%%%%%%%%%%%%%%%%%%%%%%%%%%%%%%%%%%%%%%%%%%%%%%%%%%%%%%%%%%%%%%%%%%%%%

The top-down welfare-maximizing approach does not account for the presence of an outside option that may cause low-utility agents to exit.  
By directing support primarily toward already central agents, a trickle-down strategy may leave peripheral or disadvantaged agents below the outside-utility threshold.  
These agents then exit the network despite increases in aggregate welfare, which in turn may trigger further departures and ultimately reduce both participation and overall welfare.  
This motivates a complementary perspective that focuses explicitly on stabilizing vulnerable agents rather than maximizing welfare at a single point in time.

Rather than maximizing total welfare, our immediate objective is thus to ensure that no agent’s equilibrium utility falls below the outside option $u_{\mathrm{out}}$.
This formulation captures a stabilization policy aimed at preventing further network decay by supporting the agents that are closest to leaving.
We seek the minimal additional support $\boldsymbol{y} \geqslant 0$ such that the updated incentives $\boldsymbol{\alpha}' = \boldsymbol{\alpha} + \boldsymbol{y}$ yield equilibrium utilities above the outside option for all agents:
\begin{equation}
\label{eq:stabilization_constraint}
U_i^\ast(\boldsymbol{\alpha}')
=
\frac{1}{2}
\Big(
\big[(I - \beta A)^{-1}\boldsymbol{\alpha}'\big]_i
\Big)^{2}
\;\geqslant\;
u_{\mathrm{out}}
\end{equation}
for all $i$.
The corresponding optimization problem can be written as
\begin{equation}
\label{eq:minimal_support_problem}
\begin{aligned}
&\min_{y \geqslant 0} && \mathbf{1}^\top \boldsymbol{y} \\
&\text{subject to} && (I - \beta A)^{-1}(\boldsymbol{\alpha} + \boldsymbol{y}) \;\geqslant\; \sqrt{2\,u_{\mathrm{out}}}\,\mathbf{1}.
\end{aligned}
\end{equation}
Equation~\eqref{eq:minimal_support_problem} defines a linear optimization problem with a convex feasible set.
Since both the objective and the constraints are linear, the problem can be solved efficiently with convex programming methods.

Figure~\ref{fig:support_schemes} displays the solution by illustrating how total welfare evolves under our minimal-support scheme (bottom-up approach, BU) as the outside option increases.
As expected, higher outside options require larger total support $\|\mathbf{y}^*\|$, since more incentives are needed to keep all agents above the threshold.
Nevertheless, by construction, the number of active agents in the system remains constant, and overall welfare continues to rise with additional support.
This pattern stands in sharp contrast to conventional welfare-maximizing interventions (top-down approach, TD) \citep{Galeotti2020}.
When the same total budget $\|\mathbf{y}^*\|$ is allocated to the high-centrality agents according to the TD approach, 
the resulting trickle-down effects fail to sustain the least advantaged agents.
Some agents still fall below the participation threshold, causing exits that ultimately reduce total welfare.

%%%%%%%%%%%%%%%%%%%%%%%%%%%%%%%%%%%%%%%%%%%%%%%%%%%%%%%%%%%%%%%%%%%%%%%%%%%%%%%%%%%%%%%%%%%%%%%%%%%%%%%%%%%%%%
\begin{figure}[!htb]
    \centering
    \includegraphics[width=\linewidth]{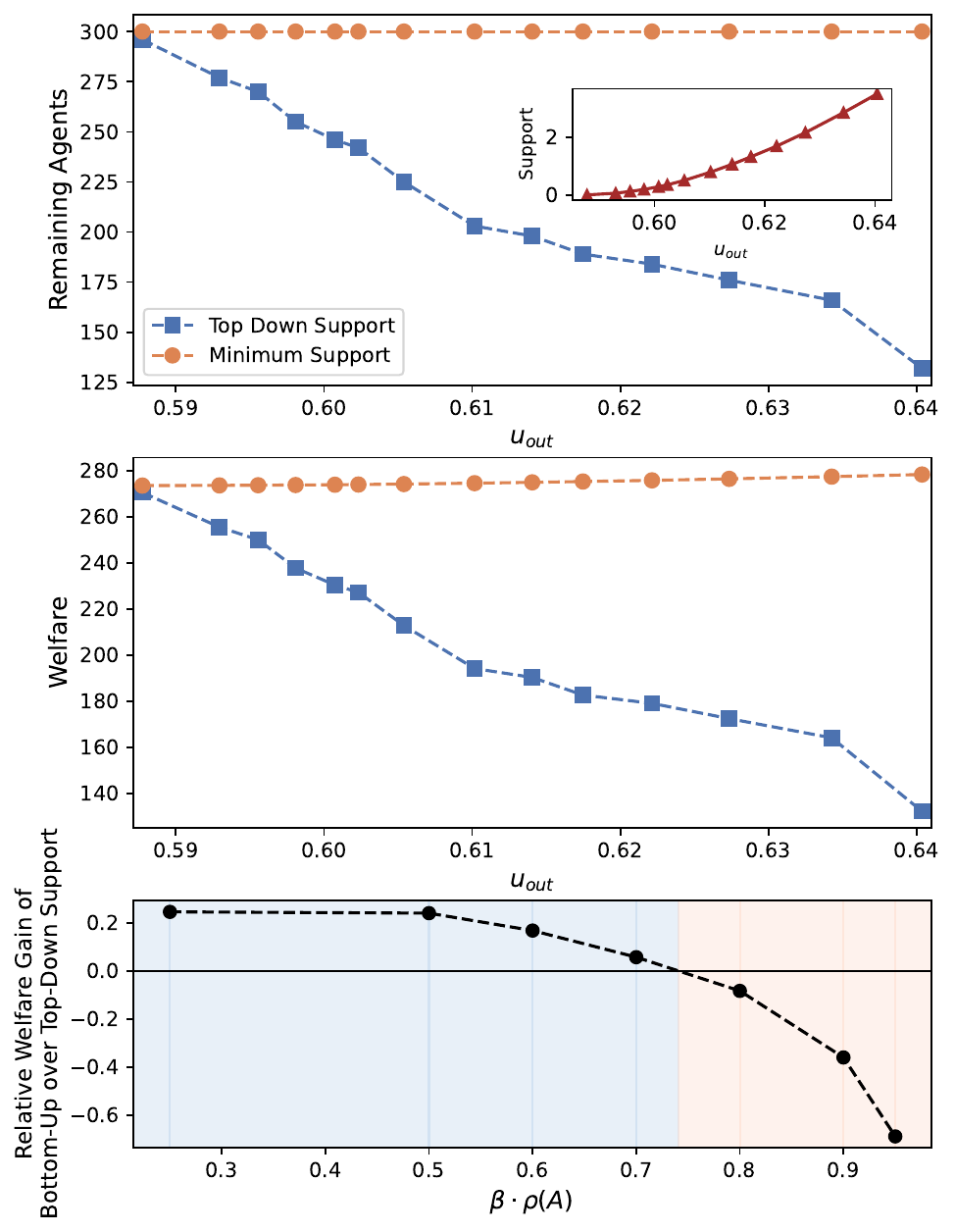}
	\caption{
    \textbf{Comparison of network support schemes.}
    We consider the same power-law cluster graph as in Figure~\ref{fig:beta_and_stochasticity_dependence} with $\beta \rho(A)=0.5$, subject to varying outside options.
    Inset: 
    Solution of Eq. \eqref{eq:minimal_support_problem}, i.e. total support $\|\mathbf{y}^*\|$ required to keep all agents active as the outside option increases.
    Top: 
    Number of active agents remaining in the network under both support schemes:
    the bottom-up scheme (BU) directly targets weak agents and maintains full network integrity.
    The welfare-maximizing top-down policy (TD), with the same budget as BU, fails to sustain participation of low-utility agents.
    Middle:  
    Evolution of total welfare under the two support schemes.
    The BU approach raises welfare even as outside options rise, 
    whereas a TD one leads to declining aggregate utility.
    Bottom: 
    Relative effectiveness of the two support schemes as a function of the spectral radius.
	}
    \label{fig:support_schemes}
\end{figure}
%%%%%%%%%%%%%%%%%%%%%%%%%%%%%%%%%%%%%%%%%%%%%%%%%%%%%%%%%%%%%%%%%%%%%%%%%%%%%%%%%%%%%%%%%%%%%%%%%%%%%%%%%%%%%%

Which approach is more effective depends critically on the spectral structure of the underlying network (Figure~\ref{fig:support_schemes}, bottom).
The spectral radius provides a natural diagnostic for choosing between a top-down, welfare-maximizing intervention and a bottom-up, vulnerability-based stabilization strategy.
When the spectral radius $\beta \rho(A)$ is large, the equilibrium map $(I - \beta A)^{-1}$ amplifies shocks strongly and distributes them broadly across the network.
In this regime, strategic complementarities are powerful: increases in the incentives of central agents propagate widely, indirectly raising the utilities of many others.  
As a result, allocating support according to the welfare-maximizing solution is highly efficient. 
In contrast, when the spectral radius is small, the network’s amplification capacity is limited and the indirect effects of central agents are weak.  
Spillovers dissipate quickly, and increases in the incentives of well-connected agents fail to lift low-utility agents above the outside option.  
In this regime, a welfare-maximizing strategy can be ineffective or even counterproductive:
it directs resources toward agents who already benefit from strong positions in the network, while those at the margin remain vulnerable to exit.  
Stabilizing the network instead requires targeted support to these vulnerable agents.

\subsubsection*{Creating Global Value in Strategic Networks}
%%%%%%%%%%%%%%%%%%%%%%%%%%%%%%%%%%%%%%%%%%%%%%%%%%%%%%%%%%%%%%%%%%%%%%%%%%%%%%%%%%%%%%%%%%%%%%%%%%%%%%%%%%%%%%

Top-down interventions are attractive because they are parsimonious and scalable: when complementarities are strong, supporting a small set of central agents can indirectly stabilize large parts of the network through spillovers.  
Their effectiveness, however, hinges on the network’s amplification capacity, summarized by the factor $\beta \rho(A)$.  
This raises a natural design question: how can a network increase its amplification capacity so that centralized support diffuses broadly and efficiently rather than remaining confined to a few well-connected agents?

One important mechanism operates through global externalities that are not reducible to local connections.  
In many socio-economic systems, the value of participation derives not only from direct interactions but from the overall quality, visibility, or legitimacy of the system as a whole.  
Examples include shared infrastructure, collective reputation, institutional credibility, or a common public-good environment.  
To capture this idea, the local interaction structure can be augmented with a global complementarity component, reflecting that agents benefit from aggregate activity in proportion to their own engagement.  
This leads to an extended utility function of the form
\begin{equation}
	\label{eq:linquad_utility_extended}
	U_i
	=
	\alpha x_i
	-
	\frac{1}{2} x_i^2
	+
	\beta x_i \sum_{j=1}^N A_{ij} x_j
	+
	\gamma x_i \left( \frac{1}{N} \sum_{j=1}^N x_j \right)
\end{equation}
where $\gamma>0$ captures the strength of the global externality.
The additional term introduces a system-level feedback: individual incentives now depend partly on the average activity of the entire network, not only on local neighborhoods.

From a spectral perspective, this global component acts as a uniform externality that increases the spectral radius of the effective interaction matrix (see SI Appendix for a proof).  
Intuitively, global feedback links every agent to every other agent, reinforcing indirect paths and increasing the system’s capacity to amplify incentives.  
As a result, the network moves toward a regime in which top-down interventions become more powerful: support targeted at central agents propagates not only through local connections but through the global channel, indirectly stabilizing peripheral and otherwise vulnerable agents.

In the context of online platforms and social networks, this highlights the importance of investments that raise the collective value of participation rather than merely strengthening local ties.  
Examples include platform-wide reputation systems, shared norms and governance structures, high-quality moderation, common identity or branding, and public-good features such as search, recommendation, or discovery tools.  
By increasing the global attractiveness of the system, such features raise the effective level of strategic complementarity and make centralized support schemes more effective at sustaining participation.

%%%%%%%%%%%%%%%%%%%%%%%%%%%%%%%%%%%%%%%%%%%%%%%%%%%%%%%%%%%%%%%%%%%%%%%%%%%%%%%%%%%%%%%%%%%%%
%%%%%%%%%%%%%%%%%%%%%%%%%%%%%%%%%%%%%%%%%%%%%%%%%%%%%%%%%%%%%%%%%%%%%%%%%%%%%%%%%%%%%%%%%%%%%
\section*{Discussion \& Future Work}
%%%%%%%%%%%%%%%%%%%%%%%%%%%%%%%%%%%%%%%%%%%%%%%%%%%%%%%%%%%%%%%%%%%%%%%%%%%%%%%%%%%%%%%%%%%%%
%%%%%%%%%%%%%%%%%%%%%%%%%%%%%%%%%%%%%%%%%%%%%%%%%%%%%%%%%%%%%%%%%%%%%%%%%%%%%%%%%%%%%%%%%%%%%

We have developed a framework for network decay driven by strategic exit in the presence of improving outside options.  
The central insight is that decay dynamics are governed by the strength of strategic complementarities, summarized by the amplification factor $\beta \rho(A)$.  
When complementarities are weak, exit follows a local threshold-cascade logic, where failures propagate through vulnerable neighborhoods.  
When complementarities are strong, departures transmit globally, producing rupture-like dynamics with abrupt collapse and limited predictive power of standard structural metrics.  
These regimes also imply different policy prescriptions: stabilizing vulnerable agents is most effective under weak complementarities, whereas targeting central agents is optimal when global feedbacks dominate.

Several directions for future work follow naturally from the limitations of our framework.  
First, our model assumes a single network with a uniform outside option and symmetric interactions between agents.  
In many socio-economic systems, however, outside opportunities arise from alternative networks, and agents can strategically relocate rather than simply exit.  
Modeling multiple competing networks would allow outside options to emerge endogenously and would naturally accommodate both entry and exit, rather than imposing the monotonic decay structure studied here.  
Such an extension would also capture settings in which highly central or productive agents leave first, for instance when they have better alternatives elsewhere.

Second, we assume that interactions are symmetric, whereas many real-world networks exhibit directional or asymmetric dependencies, such as influence hierarchies or supply-chain relationships.  
Allowing for asymmetric interaction weights would introduce richer patterns of propagation, where shocks may travel differently across directions and create new forms of fragility.

Finally, our analysis relies on an adiabatic adjustment process in which the network fully re-equilibrates after each wave of exits.  
Relaxing this assumption to allow for non-equilibrium dynamics, where agents adjust faster than the system can stabilize, may generate additional transient instabilities and path-dependent outcomes.  
Understanding how such timing frictions interact with strategic complementarities remains an important avenue for future research.

%%%%%%%%%%%%%%%%%%%%%%%%%%%%%%%%%%%%%%%%%%%%%%%%%%%%%%%%%%%%%%%%%%%%%%%%%%%%%%%%%%%%%%%%%%%%%
\section*{Acknowledgments}
%%%%%%%%%%%%%%%%%%%%%%%%%%%%%%%%%%%%%%%%%%%%%%%%%%%%%%%%%%%%%%%%%%%%%%%%%%%%%%%%%%%%%%%%%%%%%

We thank Didier Sornette, Alec Kirkley, Morgan Frank, and Seth Benzell for helpful discussions and insightful comments that substantially improved the clarity and framing of this work.

%%%%%%%%%%%%%%%%%%%%%%%%%%%%%%%%%%%%%%%%%%%%%%%%%%%%%%%%%%%%%%%%%%%%%%%%%%%%%%%%%%%%%%%%%%%%%
\bibliographystyle{unsrt}
\bibliography{bibliography.bib}
%%%%%%%%%%%%%%%%%%%%%%%%%%%%%%%%%%%%%%%%%%%%%%%%%%%%%%%%%%%%%%%%%%%%%%%%%%%%%%%%%%%%%%%%%%%%%

%%%%%%%%%%%%%%%%%%%%%%%%%%%%%%%%%%%%%%%%%%%%%%%%%%%%%%%%%%%%%%%%%%%%%%%%%%%%%%%%%%%%%%%%%%%%%
%%%%%%%%%%%%%%%%%%%%%%%%%%%%%%%%%%%%%%%%%%%%%%%%%%%%%%%%%%%%%%%%%%%%%%%%%%%%%%%%%%%%%%%%%%%%%
\section{Methods}
%%%%%%%%%%%%%%%%%%%%%%%%%%%%%%%%%%%%%%%%%%%%%%%%%%%%%%%%%%%%%%%%%%%%%%%%%%%%%%%%%%%%%%%%%%%%%
%%%%%%%%%%%%%%%%%%%%%%%%%%%%%%%%%%%%%%%%%%%%%%%%%%%%%%%%%%%%%%%%%%%%%%%%%%%%%%%%%%%%%%%%%%%%%

%%%%%%%%%%%%%%%%%%%%%%%%%%%%%%%%%%%%%%%%%%%%%%%%%%%%%%%%%%%%%%%%%%%%%%%%%%%%%%%%%%%%%%%%%%%%%
\subsection*{Equilibrium Structure of the Decay Process}
\label{sec:methods_equilibrium}
%%%%%%%%%%%%%%%%%%%%%%%%%%%%%%%%%%%%%%%%%%%%%%%%%%%%%%%%%%%%%%%%%%%%%%%%%%%%%%%%%%%%%%%%%%%%%

At a fixed outside utility level $u_{\mathrm{out}}$, let $S \subseteq \{1,\dots,N\}$ denote a subset of surviving agents and let $A[S]$ denote the adjacency matrix of the induced subnetwork on $S$.
The equilibrium activity levels on this induced network are
$
x^*(S)
=
\alpha (I-\beta A[S])^{-1}\mathbf 1_S,
$
provided that $\beta \rho(A[S]) < 1$.
The corresponding equilibrium utilities are
$
U_i^*(S)
=
\frac12 \bigl(x_i^*(S)\bigr)^2.
$

A subset $S$ is said to be \emph{stable} at outside utility $u_{\mathrm{out}}$ if
$
U_i^*(S)\ge u_{\mathrm{out}}
$
for all  $i\in S$.
Stable subnetworks need not be unique.
For example, after peripheral agents exit, a smaller dense core may remain self-sustaining even though the full network is no longer stable.

The key structural property underlying the decay process is monotonicity.
Using the Neumann expansion \eqref{eq:neumann}, 
equilibrium activities can be interpreted as weighted sums over walks on the induced network.
Adding agents to the surviving set introduces additional nonnegative walks and therefore weakly increases the equilibrium utility of all remaining agents.
As a consequence, the union of stable subnetworks is itself stable.
This implies the existence of a unique maximal stable subset at every outside-utility level,
$
S^*(u_{\mathrm{out}})
=
\bigcup_{S \in \mathcal S(u_{\mathrm{out}})} S,
$
where $\mathcal S(u_{\mathrm{out}})$ denotes the family of stable subnetworks.

The iterative removal process (Box~\ref{box:iterative-algorithm}) converges precisely to this maximal stable subset.
Starting from the full network, unstable agents are removed, equilibrium utilities are recomputed on the reduced network, and the procedure is repeated until all remaining agents satisfy the stability condition.
The monotonicity property implies that the final survivor set is independent of the order in which unstable agents are removed during a cascade.

The same structure also yields a dynamic interpretation.
Consider a repeated participation game in which agents repeatedly decide whether to remain active or irreversibly exit while outside opportunities improve over time.
At each stage, agents compare their equilibrium utility inside the network to the prevailing outside option.
Under strategic complementarities, the maximal stable subset $S^*(u_{\mathrm{out}})$ coincides with the inclusion-maximal subgame-perfect equilibrium of the corresponding dynamic game.
Hence the iterative removal process traces the largest dynamically sustainable network as external pressure increases.

Formal proofs of monotonicity, existence and uniqueness of maximal stable subnetworks, and the dynamic equilibrium characterization are provided in the Appendix.

%%%%%%%%%%%%%%%%%%%%%%%%%%%%%%%%%%%%%%%%%%%%%%%%%%%%%%%%%%%%%%%%%%%%%%%%%%%%%%%%%%%%%%%%%%%%%
\subsection*{Derivation of Local Threshold Rule}
\label{apx:threshold_derivation}
%%%%%%%%%%%%%%%%%%%%%%%%%%%%%%%%%%%%%%%%%%%%%%%%%%%%%%%%%%%%%%%%%%%%%%%%%%%%%%%%%%%%%%%%%%%%%

Suppose that a set of agents $S$ leaves the network, leaving the remaining active agents $R$.
We can then decompose the adjacency matrix into four blocks, $A_{RR}, A_{RS}, A_{SR}$ and $A_{SS}$, where $A_{SR}$ captures cross-connections between two groups $S, R \subseteq \{1, 2, \dots, N\}$, and similar for the other matrix blocks. 
The change in the equilibrium activity levels of the surviving agents due to the departure of $S$ then becomes 
\begin{equation}
	\label{eq:activity_change}
	\Delta x_R^{(-S)} = \alpha \big[(I - \beta A_{RR})^{-1} \beta A_{RS} \mathbf{1}_S \big],
\end{equation}
where $\mathbf{1}_S$ is a vector of ones over the departing agents.
This expression shows that the impact of a departure propagates along paths of all lengths through the surviving network, with longer paths discounted by higher powers of $\beta$.  

The problem simplifies considerably when $\beta$ is small relative to the network's connectivity.
Formally, when $\beta \rho(A) \ll 1$, we can approximate the inverse as
$
(I - \beta A)^{-1} = I + \beta A + O(\beta^2).
$
Substituting into the equilibrium, we obtain
\begin{equation}
	\label{eq:action_small_beta_approx}
	x_i \approx \alpha \big(1 + \beta d_i \big),
\end{equation}
where $d_i = \sum_j A_{ij}$ is the (weighted) degree of agent $i$.

Suppose a set of neighbors $S$ of agent $i$ leaves the network.
Let the total weight of the removed edges be
$
m_i = \sum_{j \in S} A_{ij}.
$
The new action of agent $i$ becomes
\begin{equation}
	x_i' \approx \alpha \big(1+\beta(d_i - m_i)\big) = x_i - \alpha \beta m_i.
\end{equation}
Equilibrium utility is given by Eq. \eqref{eq:equilibrium_utility}. 
For small changes, a first-order approximation yields
\begin{equation}
\Delta U_i = U_i' - U_i \approx x_i \Delta x_i.
\end{equation}
Substituting $\Delta x_i = -\alpha \beta m_i$ and using $x_i \approx \alpha$ to leading order, we obtain
\begin{equation}
\Delta U_i \approx -\,\alpha^2 \beta m_i.
\end{equation}
At time $t$, let the outside utility be $u_{\mathrm{out}}^{(t)}$ and define the reserve utility of agent $i$ as
\begin{equation}
r_i = U_i - u_{\mathrm{out}}^{(t)}.
\end{equation}
Agent $i$ fails when $U_i' \leqslant u_{\mathrm{out}}^{(t)}$, i.e.
\begin{equation}
U_i + \Delta U_i \leqslant u_{\mathrm{out}}^{(t)}
\quad \Longleftrightarrow \quad
\alpha^2 \beta m_i \geqslant r_i.
\end{equation}
In a weighted network, failure occurs when the total lost weight $m_i$ exceeds $r_i / (\alpha^2 \beta)$.
In an unweighted network, $m_i$ equals the number of failed neighbors.
We therefore define the integer threshold
\begin{equation}
\theta_i = \left\lceil \frac{r_i}{\alpha^2 \beta} \right\rceil,
\end{equation}
so that agent $i$ fails whenever the number of failed neighbors is greater than or equal to $\theta_i$.

%%%%%%%%%%%%%%%%%%%%%%%%%%%%%%%%%%%%%%%%%%%%%%%%%%%%%%%%%%%%%%%%%%%%%%%%%%%%%%%%%%%%%%%%%%%%%
\subsection{Empirical Data and Network Construction}
\label{apx:empirical_data}
%%%%%%%%%%%%%%%%%%%%%%%%%%%%%%%%%%%%%%%%%%%%%%%%%%%%%%%%%%%%%%%%%%%%%%%%%%%%%%%%%%%%%%%%%%%%%

\subsubsection*{Reddit interaction networks.}
Reddit data are collected using the official Reddit API (PRAW).
We consider two large subreddits, WallStreetBets and sourdough, both of which saw surges in activity during the COVID-19 period followed by sustained declines.
The dataset includes submissions and all associated comments, with timestamps and user identifiers.
We construct interaction networks by linking users who reply to one another in discussion threads.
Each comment induces a directed edge from the commenter to the author of the parent comment or submission.
We assign unit weight to each such interaction and aggregate repeated interactions over time to obtain weighted communication networks.

\subsubsection*{Stack Exchange interaction networks.}
Stack Exchange data are obtained from publicly available data dumps hosted on the Internet Archive (\url{https://archive.org/details/stackexchange}).
We focus on the \texttt{codereview} and \texttt{datascience} communities, which have experienced declining engagement in recent years, partly due to the increasing use of large language models for coding and technical assistance.
The data include all questions, answers, and comments.
We construct interaction networks by linking users who answer questions or comment on posts, creating directed edges from the responding user to the original author.
As for Reddit, each interaction receives unit weight, and repeated interactions are aggregated over time into weighted networks.

\subsubsection*{ERC20 transaction networks.}
We construct transaction networks for ERC20 tokens using on-chain transfer data obtained via the Bitquery API (\url{https://bitquery.io/}).
We focus on two projects, Dentacoin and Salt, both of which experienced substantial declines in activity following the contraction of the cryptocurrency market after the 2017--2018 boom.
The data consist of all token transfers over multi-year periods, with each transaction specifying sender, receiver, amount, and timestamp.
We interpret each transfer as a directed interaction between users and initially assign edge weights equal to the transferred token volume.
These interactions are then aggregated over time intervals to obtain weighted transaction networks.

\subsubsection*{Temporal network construction.}
For all datasets, we construct temporal networks from time-stamped interaction data by aggregating edges over fixed intervals (one week).
At each time step, we retain only users that are active within a rolling window, defined by having at least one interaction both before and after the current time.
This filters out transient participants and ensures that the observed decay reflects sustained disengagement rather than short-lived activity.
We then symmetrize the resulting networks by combining reciprocal interactions into a single undirected weighted edge and extract the largest connected component to ensure comparability across time.
While this symmetrization is consistent with the undirected formulation of our baseline model, relaxing it to allow for asymmetric interactions is a natural extension for future work.

To focus on the decay phase, we restrict attention to the period following peak activity, defined as the time at which the number of active users is maximized.
We further restrict the node set to users present at the peak, thereby obtaining a strictly decreasing sequence of induced subgraphs.
This construction aligns with the theoretical framework, which models network decay as a sequence of endogenous exits without entry of new agents.

\balance
%%%%%%%%%%%%%%%%%%%%%%%%%%%%%%%%%%%%%%%%%%%%%%%%%%%%%%%%%%%%%%%%%%%%%%%%%%%%%%%%%%%%%%%%%%%%%
\clearpage
\onecolumn
\appendix

\renewcommand{\thesection}{Appendix \Alph{section}}

\titleformat{\section}
  {\large\bfseries}
  {\thesection.}
  {0.5em}
  {}

%%%%%%%%%%%%%%%%%%%%%%%%%%%%%%%%%%%%%%%%%%%%%%%%%%%%%%%%%%%%%%%%%%%%%%%%%%%%%%%%%%%%%%%%%%%%%
\section{Equilibrium Structure of the Decay Process}
\label{apx:well_defined}
%%%%%%%%%%%%%%%%%%%%%%%%%%%%%%%%%%%%%%%%%%%%%%%%%%%%%%%%%%%%%%%%%%%%%%%%%%%%%%%%%%%%%%%%%%%%%

In this appendix, we formalize the equilibrium structure underlying the iterative removal process introduced in the main text.
We first prove a monotonicity property of equilibrium utilities under node additions.
This implies that the union of stable subnetworks is itself stable, guaranteeing the existence of a unique maximal stable survivor set at every outside-utility level.
As a consequence, the cascade outcome is independent of the order in which unstable agents are removed.
In the following appendix, we then show that this maximal stable set also admits a dynamic game-theoretic interpretation as the inclusion-maximal subgame-perfect equilibrium of an irreversible participation game.

Let $A \in \mathbb{R}^{N\times N}$ be a nonnegative adjacency matrix and let $\beta \geqslant 0$.
For any subset $S \subseteq \{1,\dots,N\}$, write $A[S]$ for the principal submatrix induced by $S$, and let
\begin{equation}
b(S)
:=
(I-\beta A[S])^{-1}\mathbf 1_S
\end{equation}
denote the corresponding Bonacich centrality vector.
A subset $S$ is stable at outside-utility level $u_{\mathrm{out}}$ if
\begin{equation}
b_i(S)\ge c
\qquad
\text{for all } i\in S,
\end{equation}
where $c=\sqrt{2u_{\mathrm{out}}}/\alpha$ follows from Eq.~\eqref{eq:equilibrium_utility}.

The following theorem establishes the key monotonicity property.

\begin{theorem}[Union stability]
\label{thm:union_stability}
Let $S_1,S_2 \subseteq \{1,\dots,N\}$ satisfy
\begin{gather}
(I - \beta\, A[S_1])^{-1} \mathbf 1_{S_1}
\;\geqslant\;
c\,\mathbf 1_{S_1},
\label{assumption:S1}
\\
(I - \beta\, A[S_2])^{-1} \mathbf 1_{S_2}
\;\geqslant\;
c\,\mathbf 1_{S_2},
\label{assumption:S2}
\end{gather}
for some $c\ge0$.
Then their union also satisfies
\begin{equation}
(I - \beta\, A[S_1 \cup S_2])^{-1} \mathbf 1_{S_1 \cup S_2}
\;\geqslant\;
c\,\mathbf 1_{S_1 \cup S_2}.
\label{conclusion:union}
\end{equation}
\end{theorem}

\begin{proof}
Using the Neumann expansion,
\[
(I - \beta A[S])^{-1}
=
\sum_{k=0}^{\infty}\beta^k A[S]^k,
\]
the $i$-th entry of
\[
(I - \beta A[S])^{-1}\mathbf 1_S
\]
can be interpreted as the weighted count of all walks starting at node $i$ and remaining inside the induced subgraph on $S$, with walks of length $k$ weighted by $\beta^k$.

Now consider $S=S_1\cup S_2$.
For any node $i\in S_1$, the set of walks contained in $S_1\cup S_2$ contains all walks contained in $S_1$.
Therefore,
\[
\big[(I-\beta A[S_1\cup S_2])^{-1}\mathbf 1_{S_1\cup S_2}\big]_i
\ge
\big[(I-\beta A[S_1])^{-1}\mathbf 1_{S_1}\big]_i.
\]
Using Eq.~\eqref{assumption:S1}, it follows that
\[
\big[(I-\beta A[S_1\cup S_2])^{-1}\mathbf 1_{S_1\cup S_2}\big]_i
\ge c
\qquad
\text{for all } i\in S_1.
\]
The same argument applied to $S_2$ using Eq.~\eqref{assumption:S2} yields the result for all $i\in S_2$.
Together, these imply Eq.~\eqref{conclusion:union}.
\end{proof}

Theorem~\ref{thm:union_stability} immediately implies the existence of a unique maximal stable subnetwork at every outside-utility level.

\begin{corollary}[Existence of a maximal stable subset]
Let $\mathcal S(u_{\mathrm{out}})$ denote the collection of stable subnetworks at outside utility $u_{\mathrm{out}}$.
Then
\begin{equation}
S^*(u_{\mathrm{out}})
=
\bigcup_{S\in\mathcal S(u_{\mathrm{out}})} S
\end{equation}
is itself stable and is therefore the unique inclusion-maximal stable subset.
\end{corollary}

This result guarantees that the iterative removal process described in the main text is well-defined.
Starting from the full network and repeatedly removing unstable agents converges to the same maximal stable subset independently of the order in which unstable agents are removed during a cascade.

%%%%%%%%%%%%%%%%%%%%%%%%%%%%%%%%%%%%%%%%%%%%%%%%%%%%%%%%%%%%%%%%%%%%%%%%%%%%%%%%%%%%%%%%%%%%%
\section{Dynamic Equilibrium Interpretation}
\label{apx:dynamic_game}
%%%%%%%%%%%%%%%%%%%%%%%%%%%%%%%%%%%%%%%%%%%%%%%%%%%%%%%%%%%%%%%%%%%%%%%%%%%%%%%%%%%%%%%%%%%%%

In this appendix, we formalize the dynamic game-theoretic interpretation of the iterative removal process introduced in the main text.
We consider a population of agents embedded in a fixed network who repeatedly decide whether to remain active or irreversibly exit while outside opportunities improve over time.
We show that the maximal stable subsets characterized in \ref{apx:well_defined} correspond to the inclusion-maximal subgame-perfect equilibrium of the resulting dynamic participation game.

\subsection{Dynamic Exit Game}
%%%%%%%%%%%%%%%%%%%%%%%%%%%%%%%%%%%%%%%%%%%%%%%%%%%%%%%%%%%%%%%%%%%%%%%%%%%%%%%%%%%%%%%%%%%%%

Let $V=\{1,\dots,N\}$ denote a finite set of agents connected through a weighted network with adjacency matrix $A\in\mathbb R_+^{N\times N}$.
For any subset $S\subseteq V$, let $A[S]$ denote the induced submatrix corresponding to the surviving agents.

Time evolves discretely over periods
\[
t=0,1,\dots,T.
\]
At each period, every currently active agent simultaneously chooses either to remain in the network or to exit permanently.
Let
\[
S^{(t)}\subseteq V
\]
denote the set of active agents at time $t$.

Agents face an outside utility
\[
u_{\mathrm{out}}^{(0)}
<
u_{\mathrm{out}}^{(1)}
<
\cdots
<
u_{\mathrm{out}}^{(T)},
\]
which increases exogenously over time.
If agent $i$ exits at time $t$, she receives the outside utility in all subsequent periods.

For a given surviving set $S^{(t)}$, agents play the linear--quadratic network game introduced in the main text.
The stage utility of agent $i\in S^{(t)}$ is
\begin{equation}
u_i(x)
=
\alpha x_i
-
\frac12 x_i^2
+
\beta x_i \sum_{j\in S^{(t)}} A_{ij}x_j.
\end{equation}
Provided
\[
\beta\rho(A[S^{(t)}])<1,
\]
the game admits the unique Nash equilibrium
\begin{equation}
x^*(S^{(t)})
=
\alpha
(I-\beta A[S^{(t)}])^{-1}\mathbf 1,
\end{equation}
with equilibrium utility
\begin{equation}
U_i^*(S^{(t)})
=
\frac12\bigl(x_i^*(S^{(t)})\bigr)^2.
\end{equation}

An agent who remains active at time $t$ therefore receives utility
\[
U_i^*(S^{(t)}),
\]
while an exited agent receives
\[
u_{\mathrm{out}}^{(t)}.
\]

\subsection{Stable Survivor Sets}
%%%%%%%%%%%%%%%%%%%%%%%%%%%%%%%%%%%%%%%%%%%%%%%%%%%%%%%%%%%%%%%%%%%%%%%%%%%%%%%%%%%%%%%%%%%%%

Given an outside utility level $u_{\mathrm{out}}$, a subset
\[
S\subseteq V
\]
is stable if every surviving agent weakly prefers remaining in the network to exiting:
\begin{equation}
U_i^*(S)
\ge
u_{\mathrm{out}}
\qquad
\text{for all } i\in S.
\label{eq:stable_condition_appendix}
\end{equation}

As discussed in Appendix~\ref{apx:well_defined}, stable subsets need not be unique.
A dense core of strongly connected agents may remain self-sustaining even after peripheral agents become unstable.
However, equilibrium utilities are monotone under node additions:
adding active agents weakly increases the utility of every survivor.
Theorem~\ref{thm:union_stability} therefore implies that the union of stable subsets remains stable.

Consequently, for every outside-utility level there exists a unique maximal stable subset
\[
S^*(u_{\mathrm{out}})
=
\bigcup_{S\in\mathcal S(u_{\mathrm{out}})} S,
\]
where $\mathcal S(u_{\mathrm{out}})$ denotes the collection of stable subnetworks.

The iterative removal process introduced in the main text computes precisely this maximal stable subset.
Starting from the full network, unstable agents are removed, utilities are recomputed on the reduced network, and the procedure is repeated until all remaining agents satisfy Eq.~\eqref{eq:stable_condition_appendix}.
By \ref{apx:well_defined}, the final survivor set is independent of the order in which unstable agents are removed.

\subsection{Dynamic Equilibrium Interpretation}
%%%%%%%%%%%%%%%%%%%%%%%%%%%%%%%%%%%%%%%%%%%%%%%%%%%%%%%%%%%%%%%%%%%%%%%%%%%%%%%%%%%%%%%%%%%%%

We now connect the iterative removal process to forward-looking strategic behavior.

Consider a dynamic participation game in which agents repeatedly decide whether to remain active or permanently exit while outside opportunities gradually improve over time.
A strategy specifies, at every time step and after every possible history of departures, whether an agent remains active or exits.

The key observation is that remaining active preserves future flexibility.
An agent whose equilibrium utility still exceeds the outside option loses nothing by remaining in the network for one more period, because exit remains available later if conditions deteriorate further.
By contrast, exit is irreversible.

This observation implies that the maximal stable subset
\[
S^*(u_{\mathrm{out}}^{(t)})
\]
constitutes a subgame-perfect equilibrium at every time step.
Agents inside the maximal stable subset prefer remaining active, while agents outside it cannot be sustained in equilibrium because their utility already falls below the prevailing outside option.

More generally, other subgame-perfect equilibria may exist.
For example, a coalition of agents may coordinate on early exit, thereby inducing additional departures.
However, no equilibrium can sustain a surviving population larger than
\[
S^*(u_{\mathrm{out}}^{(t)}).
\]
The maximal stable subset therefore defines the inclusion-maximal equilibrium trajectory of the dynamic participation game.

As a consequence, the decay dynamics analyzed in the main text can be characterized entirely through the deterministic sequence
\[
S^*(u_{\mathrm{out}}^{(0)}),
\;
S^*(u_{\mathrm{out}}^{(1)}),
\;
S^*(u_{\mathrm{out}}^{(2)}),
\;\dots
\]
without explicitly solving the full dynamic game tree.

%%%%%%%%%%%%%%%%%%%%%%%%%%%%%%%%%%%%%%%%%%%%%%%%%%%%%%%%%%%%%%%%%%%%%%%%%%%%%%%%%%%%%%%%%%%%%
\section{Condition for Global Collapse}
\label{apx:global_collapse}
%%%%%%%%%%%%%%%%%%%%%%%%%%%%%%%%%%%%%%%%%%%%%%%%%%%%%%%%%%%%%%%%%%%%%%%%%%%%%%%%%%%%%%%%%%%%%

Below, we present Theorem~\ref{thm:sufficient_condition}, which provides a sufficient condition for a global collapse in the high-$\beta$ regime.

\begin{definition}
Let $G=(V,E)$ be a graph with adjacency matrix $A$. Denote by $b(G)\in\mathbb{R}^V$ its vector of Bonacich centralities and by $\underline{b} := \min_{i\in V} b_i(G)$ the lowest centrality in the network. Define the set of \emph{rich} nodes by 
\[
R := \{i\in V : b_i(G) > \underline{b}\}
\]
and the set of poor nodes as $P := V\setminus R$. For each rich node $i\in R$, denote by  
\[
d_{\mathrm{out}}^{R}(i)
:= \bigl\lvert\{j\in P : (i,j)\in E\}\bigr\rvert
\]
the number of \emph{poor} neighbors and let $d_{\mathrm{out}}^R\in\mathbb{R}^{R}$ be the corresponding vector.
Finally, define
\[
q \coloneqq \beta~ \underline{b} ~(I - \beta A_{RR})^{-1} d_{\mathrm{out}}^{R} \in \mathbb{R}^R
\]
and denote by $q_i$ its $i$-th entry. 
\end{definition}
By Theorem~\ref{thm:union_stability}, Bonacich centrality is monotone under node and edge deletions.
This allows us to derive a sufficient condition under which increasing the outside option to the minimum equilibrium utility triggers a complete cascade, causing all nodes to leave the network.

\begin{theorem} \label{thm:sufficient_condition}
If
\[
q_i \;\ge\; b_i(G) - \underline{b} \quad \text{for all } i\in R,
\]
then for every nonempty $S\subseteq V$,
\[
\min_{i\in S} b_i(G[S]) \le \underline{b},
\]
where $b(G[S])$ denotes Bonacich centrality on the induced subgraph $G[S]$
with the same parameters $(\alpha,\beta)$.
\end{theorem}

\begin{proof}

The proof consists of four steps: 

\subsubsection*{Step 1: Poor nodes always drop out}
Take any $i\in P$, so $b_i(G)=\underline{b}$. For any induced subgraph $G[S]$ with
$i\in S$, by monotonicity we have
\[
b_i(G[S]) \le b_i(G) = \underline{b}.
\]
Therefore, if $S\cap P\neq \emptyset$, then
\[
\min_{j\in S} b_j(G[S]) \le \underline{b}.
\]
Hence it remains only to consider sets $S$ with $S\subseteq R$.

\subsubsection*{Step 2: Drop vectors}
For any nonempty $T\subseteq V$, define the drop vector
\[
\Delta^T := b_T(G) - b(G[T]) \in \mathbb{R}^T_{\ge 0}.
\]
Monotonicity implies that each $\Delta^T$ is entrywise nonnegative.
To compute $\Delta^R$ explicitly, 
partition the adjacency matrix according to $V = R\cup P$:
\[
A = \begin{pmatrix}
A_{RR} & A_{RP} \\
A_{PR} & A_{PP}
\end{pmatrix},
\qquad
b(G) = \begin{pmatrix} b_R(G) \\ b_P(G) \end{pmatrix}.
\]

The global Bonacich equation $(I - \beta A)b(G) = \alpha \mathbf{1}$ becomes
\[
\begin{pmatrix}
I - \beta A_{RR} & -\beta A_{RP} \\
-\beta A_{PR} & I - \beta A_{PP}
\end{pmatrix}
\begin{pmatrix} b_R(G) \\ b_P(G) \end{pmatrix}
=
\begin{pmatrix} \alpha \mathbf{1}_R \\ \alpha \mathbf{1}_P \end{pmatrix}.
\]
Looking at the $R$-block,
\[
(I - \beta A_{RR}) b_R(G) - \beta A_{RP} b_P(G) = \alpha \mathbf{1}_R,
\]
so
\begin{equation}
(I - \beta A_{RR}) b_R(G)
= \alpha \mathbf{1}_R + \beta A_{RP} b_P(G).
\label{block-1}
\end{equation}

For the induced subgraph $G[R]$ with adjacency matrix $A_{RR}$, the Bonacich
vector $b(G[R])\in\mathbb{R}^R$ satisfies
\[
b(G[R]) = \alpha \mathbf{1}_R + \beta A_{RR} b(G[R]),
\]
hence
\begin{equation}
(I - \beta A_{RR}) b(G[R]) = \alpha \mathbf{1}_R. \label{block-2}
\end{equation}

Subtracting \eqref{block-2} from \eqref{block-1} gives
\[
(I - \beta A_{RR}) \bigl(b_R(G) - b(G[R])\bigr)
= \beta A_{RP} b_P(G).
\]
By definition $\Delta^R = b_R(G) - b(G[R])$, so
\[
(I - \beta A_{RR}) \Delta^R = \beta A_{RP} b_P(G).
\]
Since $I - \beta A_{RR}$ is invertible,
\begin{equation}
\Delta^R
= \beta (I - \beta A_{RR})^{-1} A_{RP} b_P(G).
\label{Bonacich-equation}
\end{equation}

\subsubsection*{Step 3: Lower bound on $\Delta^R$}
By definition of $P$, every entry of $b_P(G)$ equals $\underline{b}$, so
\[
b_P(G) = \underline{b} \mathbf{1}_P.
\]
Thus
\[
A_{RP} b_P(G) = \underline{b} A_{RP} \mathbf{1}_P.
\]
Using
\eqref{Bonacich-equation}, we obtain
\[
\Delta^R
= \beta (I - \beta A_{RR})^{-1} A_{RP} b_P(G)
= \beta \underline{b} (I - \beta A_{RR})^{-1} A_{RP} \mathbf{1}_P.
\]

The vector $A_{RP}\mathbf{1}_P$ has entries
\[
(A_{RP}\mathbf{1}_P)_i = \sum_{j\in P} A_{ij}
= \bigl|\{j\in P : (i,j)\in E\}\bigr|
= d_{\mathrm{out}}^{R}(i).
\]
Hence
\begin{equation}
\Delta^R
= \beta \underline{b} (I - \beta A_{RR})^{-1} d_{\mathrm{out}}^{R}
= q.
\label{out-degree}
\end{equation}

\subsubsection*{Step 4: Drops on subsets $S\subseteq R$}
To finish the proof, we will show that there cannot exist a subset $S \subseteq R$ for which every node has Bonacich centrality strictly larger than $\underline b$.
Hence all agents will leave the network.
If $R = \emptyset$, then a cascade is automatic. Therefore, assume $S\subseteq R$ is nonempty, and fix any $i\in S$. 
The subgraph $G[S]$ is an induced subgraph of $G[R]$, so by monotonicity inside $R$,
\[
b_i(G[S]) \le b_i(G[R]).
\]
Thus
\[
\Delta^S_i
= b_i(G) - b_i(G[S])
\ge b_i(G) - b_i(G[R])
= \Delta^R_i.
\]
Combining with \eqref{out-degree},
\[
\Delta^S_i \ge \Delta^R_i = q_i.
\]
By assumption, $q_i \ge b_i(G) - \underline{b}$ for all $i\in R$. 
Hence for fixed $i\in S\subseteq R$,
\[
\Delta^S_i \ge b_i(G) - \underline{b},
\]
so
\[
b_i(G[S])
= b_i(G) - \Delta^S_i
\le b_i(G) - (b_i(G) - \underline{b})
= \underline{b}.
\]

Therefore, for every nonempty $S\subseteq R$ there exists (in fact, the sufficient condition is strong enough for this to hold for all) $i\in S$ with
$b_i(G[S])\le \underline{b}$. Combined with the case $S\cap P\neq\emptyset$ from Step 1,
this shows that for every nonempty $S\subseteq V$,
\[
\min_{i\in S} b_i(G[S]) \le \underline{b}.
\]
\end{proof}

%%%%%%%%%%%%%%%%%%%%%%%%%%%%%%%%%%%%%%%%%%%%%%%%%%%%%%%%%%%%%%%%%%%%%%%%%%%%%%%%%%%%%%%%%%%%%
\section{Global Externalities Yield Spectral Amplification}
\label{apx:global_externality}
%%%%%%%%%%%%%%%%%%%%%%%%%%%%%%%%%%%%%%%%%%%%%%%%%%%%%%%%%%%%%%%%%%%%%%%%%%%%%%%%%%%%%%%%%%%%%

Recall the extended linear-quadratic utility given by 
\begin{equation}
	\label{apx:linquad_utility_extended}
	U_i
	=
	\alpha x_i
	-
	\frac{1}{2} x_i^2
	+
	\beta x_i \sum_{j=1}^N A_{ij} x_j
	+
	\gamma x_i \left( \frac{1}{N} \sum_{j=1}^N x_j \right)
\end{equation}
where $\gamma>0$ captures the strength of the global externality.
Equation \eqref{apx:linquad_utility_extended} augments local network interactions with a global complementarity term. 
Formally, this modification is equivalent to replacing the adjacency matrix $A$ in the
baseline model by an augmented interaction matrix
\[
A^{(c)} \;=\; A + c\,\mathbf{1}\mathbf{1}^\top,
\]
where $c = \gamma/N \ge 0$ and $\mathbf{1}$ denotes the all-ones vector.
The rank-one term $c\,\mathbf{1}\mathbf{1}^\top$ captures a uniform externality through which
each agent’s effort is coupled to aggregate activity in the system.

\begin{proposition}
The spectral radius $\rho(A^{(c)})$ is non-decreasing in $c$, and strictly increasing for
$c>0$ whenever $A$ is irreducible.
\end{proposition}

\begin{proof}
Since $A^{(c)} \ge A$ entrywise for all $c \ge 0$, monotonicity of the spectral radius for
nonnegative matrices implies $\rho(A^{(c)}) \ge \rho(A)$.
If $A$ is irreducible, Perron--Frobenius theory guarantees that the leading eigenvector has
strictly positive entries, implying that the rank-one perturbation
$c\,\mathbf{1}\mathbf{1}^\top$ increases the dominant eigenvalue strictly for any $c>0$.
\end{proof}

\balance
\end{document}